\documentclass[11pt]{article}

\usepackage{fullpage}
\usepackage{authblk}
\usepackage[T1]{fontenc}
\usepackage[utf8]{inputenc}
\usepackage{ifthen}
\usepackage{hyperref}
\usepackage{amsfonts,  amsmath, amsthm} 
\usepackage[capitalise,nameinlink]{cleveref}
\usepackage{dsfont}
\usepackage{todonotes}
\usepackage{bbm}
\usepackage{mathtools}
\usepackage{verbatim}
\usepackage{multicol}
\usepackage{graphicx}
\usepackage{algorithm}
\usepackage[noend]{algpseudocode}
\floatname{algorithm}{Mechanism}

\usepackage{tikz}
\usepackage{pgfplots}

\definecolor{DarkGreen}{rgb}{0.1,0.5,0.1}
\definecolor{DarkRed}{rgb}{0.5,0.1,0.1}
\definecolor{DarkBlue}{rgb}{0.1,0.1,0.5}

\definecolor{links}{RGB}{11, 85, 255}
\definecolor{cites}{RGB}{0, 200, 0}
\definecolor{urls}{RGB}{255, 116, 0}

\def\showauthnotes{1}
\ifthenelse{\showauthnotes=1}
{
	\newcommand{\authnote}[2]{{ \footnotesize \bf{[#1: #2]~}}}
}
{ \newcommand{\authnote}[2]{} }

\newcommand{\ignore}[1]{}

\usepackage[mathscr]{euscript}
 \let\mathscr\relax

\DeclareMathOperator*{\Expectation}{\mathbb{E}}

\newcommand{\prob}[1]{\Pr\left[#1\right]}

\def\epsilon{\varepsilon}

\renewcommand{\hat}{\widehat}

\DeclareMathOperator*{\argmin}{\mathrm{argmin}}
\DeclareMathOperator*{\argmax}{\mathrm{argmax}}

\newtheorem{theorem}{Theorem}

\newtheorem{example}{Example}
\newtheorem{remark}{Remark}
\newtheorem{proposition}{Proposition}
\newtheorem{definition}{Definition}
\newtheorem{observation}{Observation}

\newtheorem{corollary}{Corollary}
\newtheorem{lemma}{Lemma}

\newcommand{\vecc}[1]{\ensuremath{\mathbf{#1}}}

\newcommand{\expect}[1]{\mathbb{E}\left[#1\right]}

\title{Tiered Mechanisms for Blockchain Transaction Fees}

\author[1]{Aggelos Kiayias}
\affil[1]{University of Edinburgh and IOG}
\affil[]{\textsf{\href{mailto:akiayias@inf.ed.ac.uk}{akiayias}@inf.ed.ac.uk}\medskip}
\author[2]{Elias Koutsoupias}
\affil[2]{University of Oxford}
\affil[]{\textsf{\href{mailto:elias@cs.ox.ac.uk}{elias}@cs.ox.ac.uk}\medskip}
\author[3]{Philip Lazos}
\author[3]{Giorgos Panagiotakos}
\affil[3]{IOG}
\affil[]{\textsf{\{\href{mailto:philip.lazos@iohk.io}{philip.lazos},
		\href{mailto:giorgos.panagiotakos@iohk.io}{giorgos.panagiotakos}\}@iohk.io}
	\medskip
	}

\begin{document}

\maketitle

\begin{abstract}
	Blockchain systems come with the promise of being inclusive for a variety of 
	decentralized applications (DApps) that can serve different purposes and 
	have different urgency requirements. Despite this, the transaction fee 
	mechanisms currently deployed in popular platforms as well as previous 
	modeling attempts for the associated mechanism design problem focus on an 
	approach that favors increasing prices in favor of those clients who value 
	immediate service during periods of congestion. 
	To address this issue, we introduce a model that captures the {\em traffic 
	diversity}  of blockchain systems and a {\em tiered pricing} mechanism that 
	is capable of implementing more inclusive transaction policies. 
	In this model, we demonstrate formally that EIP-1559, the transaction fee 
	mechanism currently used in Ethereum, is not inclusive and demonstrate 
	experimentally that its prices surge horizontally during periods of 
	congestion. 
	On the other hand, we prove formally that our mechanism achieves stable 
	prices in expectation and we provide experimental results that establish that 
	prices for transactions can be kept low for low urgency transactions, 
	resulting 
	in a diverse set of transaction types entering the blockchain. At the same 
	time, 
	perhaps surprisingly,  our mechanism does not necessarily sacrifice revenue 
	since the lowering of the prices for low urgency transactions can be covered 
	from high urgency ones due to the price discrimination ability of the 
	mechanism. 
\end{abstract}

\section{Introduction}

Blockchain systems are supposed to process a variety of transactions that come from different applications including DeFi systems, privacy mixers, atomic swaps between users who interact in online marketplaces, and many others. Transactions from different types of applications may have substantially different urgency requirements and the user's value and welfare may dramatically be affected by the time it takes for the transaction to be included in the blockchain. 
For example, consider the contrast between DeFi and settlement transactions of a layer 2 payment system. The first type naturally has a high urgency and the user's value may decrease substantially after a short time period, while the urgency of the second type may be relatively low and its users' value can be unchanged for a substantially longer period of time. Given this, it seems natural that an ideal transaction fee mechanism should accommodate the different urgency requirements of a large variety of applications. 

Nevertheless the dominant approach in blockchain systems (cf. Ethereum~\cite{Ethereum} and Bitcoin~\cite{Nakamoto2008}) as well as in modelling transaction fee mechanisms \cite{DBLP:conf/sigecom/Roughgarden21,DBLP:journals/iacr/ChungS21}
casts the process of accessing blockchain space as a  first price auction where users competing for blockchain space try to outbid each other only at the highest level urgency. 
This results in a situation where at a time of network congestion, the applications with the highest urgency and value will drive the transaction prices high --- even to a degree that low urgency applications may become too expensive to use\footnote{Both bitcoin and Ethereum average transaction charts exhibit periods of spikes that correlate with specific application use, cf. \url{https://ycharts.com/indicators/bitcoin_average_transaction_fee}
and 
\url{https://ycharts.com/indicators/ethereum_average_transaction_fee}.}. 
This state of affairs comes in sharp contrast with the promise of these systems as egalitarian workspaces where different developers and users can deploy applications and use them concurrently. 

While taking into account the urgency requirements of each transaction has many potential advantages such as better efficiency, increasing user satisfaction, and achieving higher levels of fairness by being more inclusive, it also has an important downside. A system where the users would specify their urgency requirements in detail would be cumbersome, with unnecessary high communication cost, and complicated incentives. 

Motivated by this, we put forth a simple system that balances the tradeoff between the two extremes of no urgency signaling and total urgency specifications: a \emph{tiered transaction scheme}. In such a scheme, the real estate of blockchain is divided into finitely many tiers, whose number is dynamically determined, and each one has its own service delay and price which is also dynamically determined.  Users select the tier in which they submit their transaction and pay the posted price. A rough analog of a tiered system is a multi-lane highway, in which different lanes have different speeds and gas consumption rates and users select the lane that accommodates their needs depending on their urgency and how much gas they wish to spend. As we will demonstrate  such a simple scheme can achieve substantially better congestion and welfare than a system in which all slow and fast drivers are mixed in all lanes.

Pricing systems that offer different levels of service are very common. A typical such pricing setting is the ``Fedex problem'' \cite{DBLP:conf/sigecom/FiatGKK16}, in which a few service levels are offered to the clients/users; here tiers correspond to service levels such as same-day, overnight, and 2-day delivery. In the simplest formulation of the problem, each user has a deadline $d$ and a value $v$, selects one of the offered service levels, pays the corresponding price $p$, and gets value $v$ if its package arrives by the deadline $d$; otherwise the user gets value 0. A more general formulation --- one that we consider in this work --- is that the value of a user is a decreasing function $v(d)$. Still, there are significant differences between the Fedex problem and a tiered transaction fee mechanism; for example, in a blockchain transaction fee scheme the user does not have to pay anything if the transaction is not included in the blockchain. The most important differences however is that our transaction fee mechanism should adapt dynamically to demand by changing prices and the size and number of tiers and that there is congestion due to limited capacity.

\noindent 
{\bf Our Contributions.}  In this work, we present a novel model for
capturing client demand for blockchain processing that utilizes
monotonically decreasing functions that map delays to values. Given
that the blockchain consensus protocol has to serialize transactions,
during periods of high demand there will be inevitable delays for
transactions, which accordingly reduces their utility. This gives rise
to a wide class of mechanisms, whose study in initiated here. One of
the simplest examples in this setting is the Ethereum EIP-1559
mechanism (without tips), a posted-price mechanism that focuses on
servicing all transactions immediately and disregards the dependency
of the potential decrease in value due to delays.
 
Our primary objective in settings where utility decreases in time,
which we term \emph{traffic diversity}, aims to capture the
requirement of having a policy that balances the service among
multiple types of users when the system is experiencing
congestion. The advantage of expressing this as an objective is that
we can consider transaction fee mechanisms that can conform to a
publicly shared and agreed {\em diversity policy}. A simple example of
a diversity policy would be to devote, say, 20\% of transaction space
in a block on average to low urgency transactions. 

Having a formal definition of diversity, allows us to analyze EIP-1559 and demonstrate that the set of policies it can implement is quite restricted. 
This comes at no surprise since the mechanism under
load favors a specific type of transaction - those that are willing to
pay higher transaction fees, independently of their urgency.

 \ignore{
Having a formal definition of diversity policies allows us to explore
an important class of policies that we call {\em inclusive policies.}
In an inclusive policy, transactions of sufficiently different urgency
and price profiles are recognized and acknowledged for service. A
member of the class will be parameterized by how many such distinct
types of transactions are included. We analyze EIP-1559 from an
inclusive policy perspective and we demonstrate that it does not
exhibit diversity. This comes at no surprise since the mechanism under
load favors a specific type of transaction - those that are willing to
pay higher transaction fees, independently of their urgency.
}
 
To address a rich variety of  policies, we consider
\emph{tiered pricing mechanisms}. The key feature of this class of
mechanisms is the introduction of delays in transaction
processing. Specifically each block of transactions is assigned to a
specific tier at the choice of the mechanism and may impose a
processing delay. Transactions settled in such a block will remain in
limbo until the right moment after which they will be processed and
incorporated into the ledger. If any of these transactions is not
valid at that time it will be not be included in the ledger. Tiers are
created dynamically based on demand and the price to access each tier
is determined dynamically based on the tier’s capacity.
 
We analyze the tiered mechanism from a traffic diversity perspective
and we prove that it can implement a rich variety of policies,
under
the standard assumption of ``no points masses’’ distributions and the
reasonable assumption that the expected demand is continuous. An
important and desirable property that we establish is that --- for
fixed tier sizes --- the mechanism can always find stable and unique
delays and prices that implement the desired policy.
 
We couple the theoretical conclusion with a set of experiments that
aim to examine how the tiered mechanism converges to a set of prices
and how it behaves when demand fluctuates. In particular, we study the
important setting when demand surges exceeding network capacity for a
certain period of time. We compare our mechanism’s behavior against
the EIP-1559 protocol. Our experiments demonstrate that prices for
transactions whose value does not significantly deteriorate in time
can be kept low, resulting in a diverse set of transaction types
entering the blockchain with suitable price discrimination based on
urgency. In contrast, prices in EIP-1559 surge horizontally throughout
the network congestion period. We also examine the impact on revenue
that our system has and we demonstrate that tiered pricing does not
need to sacrifice revenue in order to accommodate diversity: due to
its ability to adjust prices per tier, higher transaction fees payed
in the high urgency ``rich’’ category compensate for the lower prices
in the slower tiers.

\subsection{Related Work}
The most closely related work to our paper is the original EIP-1559 mechanism 
proposed in \cite{buterinethereum} and 
\cite{DBLP:conf/sigecom/Roughgarden21} which provided a framework for its 
analysis and proposes alternatives. Significant later results include 
\cite{chung2023foundations} which show that in general, no 
transaction fee mechanism can satisfy user incentive compatibility, miner 
incentive compatibility and be off-chain agreement proof all at once. They 
propose a change to the model incorporating potential future losses from 
misreporting and design a mechanism in this framework that can do all three, 
at the expense of using more block space for the same throughput. 
\cite{ferreira2021dynamic} propose another variant of EIP-1559, which is 
proven to be more `stable', by showing 
that the prices exhibit a martingale property for unchanging stochastic 
demand. The stability of EIP-1559 has been studied through the lens of 
dynamical systems in a string of papers 
\cite{reijsbergen2021transaction, leonardos2021dynamical, 
	leonardos2022optimality}, which show that even though the prices can show 
chaotic behaviour, the block sizes on average are indeed very close to the 
target. This can also be viewed as a theoretical justification of our model, which 
directly assumes this situation is true. There have been many other proposals 
for updating transaction fee mechanisms, aimed both at the EIP-1559 and 
Bitcoin's first price auction. One of the first significant attempts to revamp 
Bitcoin's mechanism is \cite{lavi2022redesigning}. 
They try to mitigate both the low revenue generated when congestion is low and 
truthfulness at once and propose two mechanisms, the Monopolistic Price, and 
Random Sampling Optimal Price (RSOP), the latter of which was initially 
proposed for digital goods in \cite{goldberg2006competitive}. The Monopolistic 
Price mechanism is not always truthful from the users' side, but as they show 
experimentally and \cite{yao2018incentive} proves rigorously in it is 
approximately incentive compatible when demand is high. Finally, 
\cite{buterinforward,basu2019towards} propose variants of `pay-forward' 
mechanisms, where fees do not directly go to the miner responsible for the next 
block but are distributed to others as well. The Ethereum fee burning rule could 
be seen as a culmination of this line of reasoning.

\section{Preliminaries}
We are interested in the steady state behavior of a blockchain with a 
transaction mechanism that has the additional power to 
add delays to transactions in addition to fees. Moreover, we focus on the case 
with sustained demand \emph{exceeding} throughput.

We assume that the transaction mechanism is strategyproof for myopic agents (i.e., agents that do not take into account what could happed at \emph{later} blocks and maximise their utility greedily). Therefore, we can conceptualize the \emph{steady state} interaction of users with the 
blockchain as follows: at some random point during the day  the user needs to submit a transaction. If the user truthfully report her needs, what kind of service will 
be provided to her?

\subsection{Steady State Mechanisms for Blockchains}
The transactions generated by the users are characterized by having a value 
that is dependent on the \emph{delay}. In particular, each transaction $t_i$ 
is associated with a function $v_i : \mathbb{R}_{\ge 1} \rightarrow  
\mathbb{R}_{\ge 0}$, mapping delays to values.  Notice that the minimum 
delay is 1, for a transaction included in the very next block. Naturally, we also have that if $d < d'$ 
then $v_i(d) > v_i(d')$, meaning that users prefer faster service. For convenience, we will refer to the value (or utility) of a transaction and the user who issued it interchangeably.

We use $B$ to denote the throughput of the blockchain, measured in transactions per block. 
The value functions $v_i$ of transactions trying to enter the blockchain are independent, 
identically distributed and follow the distribution 
$F$, arriving at a rate of $n$ per block.
To model high, sustained demand we assume that the average arrival rate of 
transactions is $n > B$; w.l.o.g., $n$ is assumed to be a function of $B$.
The blockchain implements a transaction 
mechanism $\mathcal{M}(v_i \;|\; n,F) = (d, p)$, mapping transactions to 
price and 
delay pairs; we will omit $n$ when it is clear from the context.\footnote{Notice that the mechanism $\mathcal{M}$ also takes into account $F$. In practice, this is not a direct input but a result of some dynamic price update mechanism that only observes some of the demand.}
Given a transaction with value $v$, its utility is $u(v \;|\; \mathcal{M}, n, F) = v(d) - p$. The \emph{demand} is defined as
$
D = \{ v \in \text{supp}(F) \;|\; v(d) - p \ge 0 \text{ for } (d, p) = 
\mathcal{M}(v \;|\; n, F)\}.
$
Transactions with non-negative utility should arrive at a rate lower than the 
throughput, thus we need that $n \cdot \prob{D} \le B$. Given that $n$ and 
$B$ are constant, we will omit them when they are clear from the context.
\begin{remark}
So far, $\mathcal{M}$ has been presented a \emph{direct} revelation mechanism, receiving as input the value function $v$ and providing some quality of service. This is better for the sake of comparing mechanisms against each other under a unified notation. In practice, the mechanisms used (and the the main result of this work) offer the same menu of choices to every user (at a given time), letting them chose without revealing their $v$. By the revelation principle, of course these mechanisms could be converted to our setting where the optimal choice of service among that menu would be done internally by $\mathcal{M}$, given $v$ as a input.
\end{remark}

\begin{remark}
	Notice that $\mathcal{M}$ is agnostic about the identity of each user since in 
	practice anonymity and fairness are ranked quite high among 
	desirable blockchain properties.
\end{remark}

We illustrate one application of the previous definitions using Ethereum as an example.
\begin{example}[Ethereum under EIP-1559]
	For myopic users, EIP-1559 (Mechanism~\ref{alg:EIP}) appears to be a simple posted price mechanism to get immediate service.

\begin{algorithm}[h] 
	\caption {The price update rule of EIP-1559 is parameterized by the 
		target load $targetLoad$. It takes as input the price of the previous block 
		$p$ and its fullness level $f$, and outputs the new price.}
\label{alg:EIP}
\begin{algorithmic}[1]			
	\Function{Update-EIP1559}{$p,f$}
	\State return $p \cdot (1+\frac{1}{8} \frac{f - targetLoad}{targetLoad})$
	\EndFunction
\end{algorithmic}		
\end{algorithm}

	To satisfy the demand constraints, the price 
	provided by the mechanism has to be $p_0$ such that
	$$
	B = n\cdot \prob{D}  = n\cdot \prob{v(0) \ge p_0}.
	$$
	Then, the mechanism would be $\mathcal{M}(v \;|\;  F) = (0, p_0)$ for all 
	$v \in F$. Theoretically, this pricing mechanism does not have convergent or 
	particularly stable prices even with unchanging stochastic demand (e.g., 
	\cite{ferreira2021dynamic} show that the price is not a martingale, but they 
	propose a modified mechanism that is). Their long-term average however, 
	should be equal to this calculated amount.
\end{example}

We require that mechanism $\mathcal{M}$ is myopically truthful at the steady state. Essentially, this means that under the steady state achieved by the truthful participation of all other users, a user who is not trying to `time' their participation (i.e., wait for a block where the price level randomly fluctuates from the average) would maximize their utility by truthfully reporting their value $v$.
\begin{definition}[Steady State Myopic Truthfulness]
	A mechanism $\mathcal{M}$ is truthful if for all $n, F, v$ and $v' \neq v$, if it holds that
	$$
		v(d) - p \ge v(d') - p',
	$$
	where 
	$(d, p) = \mathcal{M}(v \;|\; F)$ and $(d', p') = \mathcal{M}(v' \; |\; F)$
\end{definition}
In fact, a straightforward characterization further supports the study of 
mechanisms explicitly consisting of delay/price pair menus.
\begin{proposition}[Steady State Myopic Truthfulness Characterization]
	A mechanism $\mathcal{M}$ is truthful if and only if for all $n, F, v_1$ and 
	$v_2$, it holds that $d_1 \ge d_2 \Rightarrow p_1 \le p_2$ for  
	$\mathcal{M}(v_1 \;|\; F) = (d_1, p_1)$ and $\mathcal{M}(v_2 \; |\; F) = 
	(d_2, p_2)$. Additionally, if $d_1 
	=  d_2$ then $p_1 = p_2$.	
\end{proposition}


\section{Traffic Diversity}



Next, we initiate the study of  diversity in the context of blockchain transaction fee mechanisms. Intuitively, a diverse mechanism should enable transactions with different urgency profiles to make effective use of the platform. As already mentioned, current approaches  do not specifically address this issue, and tend to focus on accommodating transactions with specific urgency profiles. We start by providing a formal definition of diversity.


A number of challenges must be addressed to properly define a notion of diversity. First, deciding which applications should be served and at what rate may depend on external factors. To avoid talking about specific policies of our taste, we choose to abstract this part of the definition by introducing the notion of a diversity policy that our mechanism should implement. The exact choice of this policy is left to the mechanism designer. 

Secondly, we have to choose a vocabulary for phrasing different policies. 
In favor of simplicity, we differentiate among urgency levels by checking whether the value of a transaction   remains above a threshold value when appropriately delayed.   
The intuition is that  the value of high urgency transactions is going to drop quickly when delay increases compared to low urgency transactions. Thus, talking about the value of transactions when delayed is sufficient for our purposes.
In more detail, a diversity policy consists of a set of clauses describing the percentage of transactions ($a_i$) included in the blockchain whose valuation for a certain delay ($d_i$) exceeds some threshold ($p_i$).


\begin{definition}[Diversity policy]
	A diversity policy $P$ is a set $\{(a_i,d_i,p_i)\}_{i\in\mathbb{N}}$, where
	$a_i,d_i,p_i$ in $(0,1],\mathbb{R}_{\geq 1},(0,1]$, respectively. Triplets $(a_i,d_i,p_i)$ in $P$ are called policy clauses.
	
\end{definition}

Besides transactions with different urgency profiles making it to the blockchain, we also need to ensure that they get a sufficiently high quality of service.
Otherwise, our mechanism  may look diverse, but in reality offer very bad quality of service to certain urgency profiles.
This is captured by requiring that the utility of the transactions meeting some policy clause is greater that the utility implied by the clause. 
Finally, if all the clauses of a diversity policy are met when looking at the steady state behavior of the blockchain, we say that the mechanism implements this policy.

\begin{definition}[Implementation] 
	A mechanism $\mathcal{M}$  implements a diversity policy $P := \{(a_i,d_i,p_i)\}_{i \in \mathbb{N}}$ under load $F$  iff for all $i \in \mathbb{N}$ it holds that 
	$$\Expectation \big [|\{ j | u(v_j|\mathcal{M},F) \geq  v_j(d_i) - p_i > 0, \text{ where } (v_1,\ldots,v_n)\leftarrow F \}| \big ] \geq a_i \cdot B$$
\end{definition}


\begin{example}
	Given $F$, the EIP-1559 mechanism finds the value $p$ such that on expectation exactly $B$ transactions have value at least $p$ when delay is equal to $1$. Hence, it implements the policy $P := \{(1,1,p)\}$ under load $F$ (assuming $p$ exists).
\end{example}

\ignore{
It is not always possible to implement a diversity policy, for the simple reason that the load $F$ may not contain (on expectation) enough transactions meeting the policy constraints. 

\begin{definition}
	Load $F$ is 
\end{definition}
}

\section{A Diverse Mechanism}

In this section, we present the Tiered Pricing mechanism. This mechanism can be parameterized to implement different diversity policies based on the goals set by the system designer. As a warm-up, we first argue why EIP-1559 is a bad choice for when diversity is a desired feature of the system.

\subsection{EIP-1559 is not Diverse}

As mentioned earlier, we would like the mechanism designer to be able to select transactions with different urgency profiles to fill the blockchain. In the simplest non-trivial case, the designer would like the mechanism to satisfy a policy $P$ that includes some policy clause $(a,d,p)$, where $a>0,d>1$. We show next, that given $(a,d,p)$, there exists a simple load distribution $F$ with two urgency levels  (see Figure~\ref{fig:badload}) for which EIP-1559 does not implement any policy $P$ where $(a,d,p)\in P$. The main reason being that  EIP-1559 selects which transactions to add to the blockchain based on the value of immediate inclusion. Hence, when the network is congested, it will always select transactions with high urgency and never serve medium or low urgency transactions.

\begin{definition}\label{def:badload}
	Given parameters $a,d,p$ such that $d >1, a,p >0$, let load $F$ be defined as follows:
	\begin{itemize}
		\item with probability $1/2$, $F$ outputs a valuation function $v$ where $v(1) := 2 p$ and $v(d) := 0.2 \cdot p$,
		\item with probability $1/2$, $F$ outputs a valuation function $v$ where $v(1) := 3p/2$ and $v(d) := p$
	\end{itemize}
	and all other values of the two functions are chosen so that they are monotonic.
\end{definition}

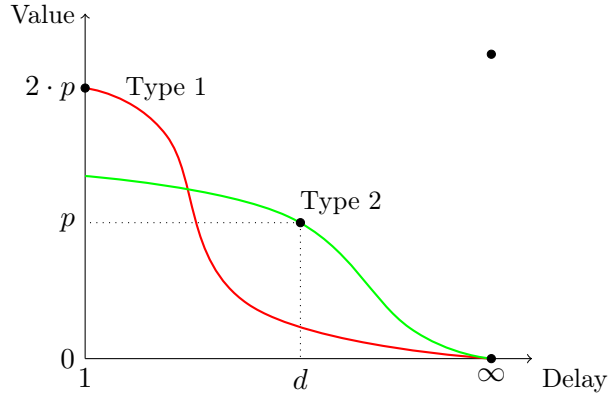
\begin{figure}[h]
	\centering
	\begin{tikzpicture}[scale=2.7, xscale = 2, yscale=1, domain=0:1, range=0:1, 
	samples=400, smooth]
		
		\draw[->] (0, 0) node[left] {$0$} -- (1.1, 0) node[below right] {\small 
		Delay};
		
		\draw[->] (0, 0) node[below] {$1$} -- (0, 1.7) node[left] {\small Value};

		\draw [color=red, thick] plot [smooth, tension = 0.7] coordinates { 
			(0, 4/3) 
			(0.2, 1.1)
			(0.4, 0.27)
			(1, 0)};
		\node at (0.2, 4/3) {\small Type 1};
		
		\draw [color=green, thick] plot [smooth, tension = 0.7] coordinates { 
			(0,0.9) 
			(0.5, 0.7) 
			(0.8, 0.15) 
			(1, 0)};
		\node at (0.63, 0.77) {\small Type 2};
		
		\draw[dotted] (0.53, 0.67) -- (0, 0.67) node[left] {$p$};
		\node [scale=0.3, draw, circle, fill] at (0.53, 0.67) {};
		\draw[dotted] (0.53, 0.67) -- (0.53, 0) node[below] {$d$};

		\node [scale=1, left] at (0, 4/3) {$2\cdot p$};
		\node [scale=0.3, draw, circle, fill] at (1, 1.5) {};
		\node [scale=0.3, draw, circle, fill] at (0, 4/3) {};   
				
		\node [scale=0.3, draw, circle, fill] at (1, 0) {};
		\node[below] at (1,0) {$\infty$};
		
	\end{tikzpicture}
	
	\caption{The distribution of Definition~\ref{def:badload}. Note, that 
	EIP-1559 will only select transactions of the first type.}
	\label{fig:badload}
\end{figure}

\begin{proposition}\label{th:EIPnotdiverse}	 
	Let $(a,d,p)$ be a policy clause where $d>1,a,p>0$.
	EIP-1559 does not implement any policy $P$ containing $(a,d,p)$ under the load $F$ of Definition~\ref{def:badload}.
\end{proposition}
\begin{proof}
	Due to the congestion assumption $n > B$ (i.e, that the number of value 
	functions sampled from $F$ is a lot larger that the available throughput 
	$B$), it follows 
	that EIP-1559 will have to set the cutoff price to closer $2\cdot p$ for large 
	enough $n$. Thus, no transactions of the 
	second type will make it to the blockchain. Given that transactions of the first 
	type do not meet the policy clause $(a,d,p)$, the theorem follows. 
\end{proof}

We have thus argued that EIP-1559 cannot implement non-trivial diversity policies even when the load has a simple form.~\footnote{Note, that the load of Definition~\ref{def:badload} is not artificially ``unimplementable'', as it contains on expectation enough transactions to satisfy the relevant policy. In fact, the mechanisms we present next will be able to deal with this type of policies and loads.}

\subsection{The Tiered Pricing Mechanism}

Next, we describe the Tiered Pricing mechanism.  Our mechanism separates available space in each block into a maximum  of $k$ different tiers. Each tier has its own delay $(d_i)_{i\in[k]}$ and price $(p_i)_{i\in[k]}$ parameters, with delays increasing and prices decreasing in successive tiers. The delay of the first tier is always set to 1, i.e., it is minimal, while all other delays and prices change dynamically respecting a set of parametric constraints. The general goal of the mechanism is to try to select minimal constraint-respecting delays and prices, while at the same time trying to maximize the number of tiers and their usage.

To be able to serve users with different urgency profiles, and thus implement the relevant policies, the delays and prices of successive tiers are related in a parameterizable way:
\begin{align} \label{ineq:TP}
	d_{j+1} \ge \lambda_j \cdot d_j \text{ for } \lambda_j > 1,
	\text{ and } p_{j+1} \le \mu_j \cdot p_j \text{ for } \mu_j < 1
\end{align}
where $(\lambda_j,\mu_j)_{j\in[k-1]}$ are parameters of the mechanism.
Any time these constraints are not met, the mechanism changes one or more of the tier parameters (delays,prices, number of tiers) in an effort to satisfy them. 
In its initial state the mechanism has a single active tier of size $B$.

Every time a new block is generated, parameters are subject to change.
Similarly to EIP-1559, the price of each (active) tier is adjusted based on the fullness of this tier in the previous block.  
The delay of different tiers changes less often (every $dFreq$ blocks), allowing some time for prices to catch up on changes on the load and stabilize.
Delays change based on the relation of successive tier prices; if $p_{i+1} > \mu_i \cdot p_{i}$, the delay of tier $i+1$ is increased to try to  push $p_{i+1}$ down to meet the relevant constraint. 
Otherwise, if $d_{i+1} \ge \lambda_i \cdot d_i$,  the delay of tier $i+1$ is decreased by $1$ with some probability $(pDecrease)$ to increase the overall efficiency of the system by avoiding possibly unnecessary delays. 
Even less often (every $tFreq$ blocks), the number of tiers may change. Namely, if the value of the last tier is greater than $addTierPrice$ and the number of active tiers $m$ is less than $k$, a new active tier is added with size $a_{m+1}\cdot B$, price $newTierPrice$ and delay $\lambda_m \cdot d_{m}$. The space taken from tier $m+1$ is reduced from tier $1$. In a similar fashion, if the price of last tier is less than $removeTierPrice$ and there are more than one active tiers, then last tier is deleted and the space used returns to tier $1$. The rationale is that if the price of the last tier is very small, then there is insignificant demand for this urgency profile, and thus it makes sense to remove the tier altogether. On the other hand, if the price of the last tier is sufficiently high, there is enough demand for this urgency profile and it makes sense to spawn a new tier to see if there is demand for even less urgent transactions. 
For a formal description of the parameter update process we point to Mechanism~\ref{alg:TP}.

Finally, in order for our mechanism to be effective, the delay penalty associated with a tier must be different than the general delay of the system, otherwise users would just pay to be included in the cheapest tier and lie about their preferences. The delay penalty is enforced by delaying the processing of transactions included in a tier by the tier's delay, i.e., for all ledger purposes a transaction is ignored until a block is produced whose timestamp exceeds that of the inclusion block by the delay amount. Accordingly, transactions that are added to a tier pay a fee equal to the tier's price.

\begin{algorithm}[H] 
		\caption {\small The Tiered Pricing parameter update function is parameterized by the maximum number of tiers $k$, the size,delay and price constraint parameters $(a_i,\lambda_i,\mu_i)_{i\in[k]}$, the total size of the block $B$, the delay and number of tier update frequencies $dFreq,tFreq$, the probability delays are decreased $pDecrease$, the price assigned to new tiers $newTierPrice$, the price bellow which the last tier is removed $removeTierPrice$. It takes as input the parameters of the previous block, the level of fullness of each tier in the previous block $(f_i)_i$, and the times the parameters have been updated so far $t$. It outputs the new set of parameters which consist of the number of tiers and their respective sizes, delays and prices.} 
		\label{alg:TP}
		\begin{algorithmic}[1]
			\small			
			\Function{updateTierParameters}{$m,((B_i,d_i,p_i))_{i\in [m]}, (f_i)_{i\in[k]}, t$}
			\State $(p'_i)_{i\in[m]} :=$ $($\Call{Update-EIP1559}{$p_i,f_i$}$)_{i\in[m]}$ \Comment{Update prices as in Mechanism~\ref{alg:EIP}}
			\State $(B'_i,d'_i)_{i\in[m]} := (B_i,d_i)_{i\in[m]}$ 
			\If{ $(t\mod dFreq = 0 )$}	\Comment{Update delays}
			\State $(d'_i)_{i\in[m]} :=$ \Call{UpdateDelays}{$m,((d_i,p'_i))_{i\in[m]}$}		
			\EndIf
			\If{ $(t\mod tFreq = 0 )$} \Comment{Update the tier number and sizes}
				\State $(m',B'_1,B'_m) :=$ \Call{UpdateTierSizes}{$m, B_1,B_m, p'_{m}$}		
				\If{ $(m'>m)$} 	\Comment{Set the delay/price of the new tier}
					\State $p'_{m'} := newTierPrice$
					\State $d'_{m'} := \lambda_{m'-1} d'_{m'-1}$
				\EndIf
			\EndIf
			
			\State return $(m',((B'_i,d'_i,p'_i))_{i\in [m']})$ \Comment{The new parameter set}
			\EndFunction
			
			\Statex
			\Function{UpdateDelays}{$m,(d_i,p'_i)_{i\in[m]}$}
			\For{$i:=1$ to $m-1$}
				\If{$(p'_{i+1} > \mu_i \cdot p'_{i})$} \Comment{If prices are too close, increase delay}
 				\State $d'_{i+1} := d_{i+1}+1$
				\Else  \Comment{Otherwise, decrease delay}
				\State with probability $pDecrease$ do $d'_{i+1} := d_{i+1} - 1$ 
				\State $d'_{i+1} := \max\{d'_{i+1}, \lambda_{i} \cdot d'_{i}  \}$ \Comment{Ensure delay constraint is satisfied}
				\EndIf	
			\EndFor		
			\State return $(d'_i)_i$
			\EndFunction

			\Statex
			\Function{UpdateTierSizes}{$m, B_1,B_m, p'_{m}$}
			\If{($m> 1$ and $p'_{m} < removeTierPrice$)} \Comment{Remove last tier}
			\State $m' := m - 1$
			\State $B'_m := 0$
			\State $B'_1 := B_1 + B_m$
			\EndIf
			\If{$(m < k$ and $p'_{m} > addTierPrice)$} \Comment{Add new tier}
			\State $m' :=  m+1$
			\State $B'_{m'} := a_{m'}$
			\State $B'_1 := B_1 - a_{m'}$ 
			\EndIf
			\State return $(m',B'_1, B'_m)$
			\EndFunction
			
		\end{algorithmic}		
	\end{algorithm}

As observed earlier, the Tiered Pricing mechanism  tries to implement policies whose clauses satisfy Inequality~\ref{ineq:TP}, i.e., the delay and prices of successive clauses satisfy the relevant inequalities. As these inequalities are parametric, the mechanism designer is given the freedom to select the relevant parameters to meet the system's needs/requirements. In the next sections, we provide a first formal analysis of Tiered Pricing and show among others that for any such selection of parameters $k,(a_i,\lambda_i,\mu_i)_i$, for a large class of load distributions, there exist delays and prices that meet Inequality~\ref{ineq:TP} and on which the system will converge in the medium-term~\footnote{In the medium-term only the  prices and delays of tiers change, as the tier number and sizes are updated less frequently.}. Moreover, the selected delays satisfy a notion of Pareto-optimality, as decreasing any of the tier delays would change prices so that the tier constrains are not satisfied anymore.
Following a similar set of arguments, our results can be extended to show that the long-term behavior of Tiered Pricing is optimal, in the sense that the mechanism converges to a set of parameters that maximize the number of tiers, while minimizing the delay parameters. 

\vspace{0.1cm}

We proceed to formally analyze the Tiered Pricing mechanism in the next section.

\section{Tiered Pricing is Diverse}
It is quite challenging to study the entire tiered pricing mechanism, so we break 
the analysis up into discrete modules. We first adapt our steady state model to 
this particular mechanism. Then, we show that by fixing the tier sizes and 
delays we can show that there always exists a unique `average' price that 
satisfy our target tier sizes. Given this property of the price update inner loop, 
we can 
proceed to the less frequent delay updates of the tiered mechanism. 
Finally, we show that such delays satisfying Inequality~\ref{ineq:TP} exist (with 
accompanying prices) for any set of parameters $k,(a_i,\lambda_i,\mu_i)_{i}$. 

\subsection{Tiered Transaction Mechanisms}
We focus on a particular case of transaction mechanisms with prices and 
delays. Specifically, we split up the blockchain throughput into \emph{tiers}, 
each of which offers a distinct price and delay option.
Suppose that our blockchain has tiers $\{0, 1, 2, \ldots, k\}$ with
corresponding sizes $\vecc B = (B_1, B_2, \ldots, B_k)$ and delays $\vecc d = 
(d_1,  d_2, \ldots , d_k)$ 
such that:
\begin{equation}
	B = \sum_{j \in [k]} B_j.
\end{equation}
As before, we refer to $B$ as the \emph{throughput} of the blockchain.  Since not all transactions will be 
included, tier $0$ is special and is reserved for those. This special tier has $B_0 
= \infty$. Tiers with larger indices offer a  reduced level of service and have 
strictly higher delays (i.e., $d_{j+1} > d_{j}$).

Since we study the steady state, let $n \ge B$ be the number of new 
transactions $T = \{t_1, t_2, \ldots, t_n\}$ generated, which will be filtered by 
the prices in order to fill the blocks. Each tier is associated with a price $\vecc p 
= (p_1, p_2, \ldots, p_k)$ which is 
needed to define the \emph{utility} of each transaction. For $t_i$, the utility 
generated if included in tier $j$ is:
\begin{equation}\label{def:trx_utility}
	u(i, j) =
	\begin{cases}
		v_i(d_j) - p_j &\text{ if } 0< j \le k\\
		0 &\text{ if }  j  = 0
	\end{cases}.
\end{equation}
Notice that the utility reflects that each transaction \emph{only} generates 
value if it is included before its corresponding deadline. Otherwise, it is better to 
be excluded. Abusing the notation, we will refer to the utility of transaction $i$ 
given delay $d$ and price $p$ as $u_i(d, p).$ We use the tuple $(\vecc B, \vecc 
d, \vecc p)$ as a shorthand to refer to a particular blockchain.\\


The tiered transaction mechanism $\mathcal{M}$ would allocate each 
transaction to the tier which maximises its utility. Since two tiers 
could offer the same utility (e.g., by having higher delays at lower prices), there 
could be a tie for the optimal tier. In this unlikely event (which has zero 
probability of happening given our distributional assumptions 
\Cref{def:demand_regular} introduced later)
transactions are allocated to utility equivalent tiers randomly. Let $X_{ij}$ be 
the 
random variable indicating that transaction $i$ chose tier 
$j$.  Given the random set of `desired' utility maximizing tiers for transaction 
$i$
\begin{equation}
	D_{i} = \{j \in \{0\} \cup [k] \;|\; u(t_i, j) \ge u(t_i, j') \quad \forall j' \},
\end{equation}
we can define the distribution of $X_{ij}$:
\begin{equation}
	X_{ij} =
	\begin{cases}1 \text{ with probability }
		\frac{B_j}{\sum_{j \in D_i} B_j} &\text{ if } j \in D_{i}\\
		0 &\text{otherwise}
	\end{cases}.
\end{equation}
Having defined how the transactions select tiers, we define the \emph{demand} 
(or \emph{used 
	space})  of 
each tier as 
\begin{equation}\label{eq:utility_maximizer}
	T_j = \sum_{i \in [n]} X_{ij}.
\end{equation}
The expected demand should be at most as much as the allocated size the 
tier.
\begin{equation}\label{eq:expected_size}
	\expect{T_j} \le B_j,
\end{equation}
except for $T_0$, which needs to capture every transaction that's left out: 
\begin{equation}\label{eq:expected_size_rejected}
	\expect{T_0} = n - \sum_{j \in [k]} \expect{T_j}.
\end{equation}
Note that Equation \ref{eq:expected_size_rejected} is always satisfied by the 
definition of $X_{ij}$.

\subsection{Distributional Assumptions}
To prove the existence results, we need to impose certain mild regularity 
assumptions on the distribution $F$. Specifically, that the expected demand is 
continuous and behaves `predictably' as the delay and price change.


\begin{definition}[Demand Regular]\label{def:demand_regular}
	The value function distribution $F$ is demand regular if the valuations 
	generated are parameterized by a tuple $(v_0, h)$ where $v_0 \in [0,1]$ is 
	the value for no delay and $h : \mathbb{R}_{\ge 1} \rightarrow [0,1]$ is the 
	discount factor and have the following form
	$$
	v(d \; | \; v_0, h) = v_0 \cdot h(d).
	$$
	Moreover, we need that
	\begin{itemize}
		\item $h(d)$ is continuous and strictly decreasing in $d$.
		\item $\lim_{d \rightarrow \infty} h(d) = 0$.
		\item The probability density function of the marginal 
		distribution of $v_0$, defined as 
		$$f(v_0 | h) = \frac{d \prob{v \le v_0 | h}}{dv_0}$$
		exists for all $h$.
		\item $f(v_0 | h) > 0$ for $v_0 \in [0,1]$.
	\end{itemize}
\end{definition}
\begin{remark}
	This assumption is similar to the more common 'no point masses' assumption 
	for single dimensional valuations and quasilinear utilities.
\end{remark}

Both conditions are necessary to ensure that a wide range of transaction fee 
mechanism designs have a well defined steady state behavior. For instance, 
notice that without the positive pdf requirement for $v_0$ it is impossible to 
guarantee that stable prices exists, even for Ethereum with just one tier. This is 
clear if the demand exceeds the supply and everyone has exactly the same 
valuation. However, the other requirements are also crucial. An example 
highlighting the need for strictly decreasing temporal 
discounts can be found in \Cref{obs:no_price}.

It is not always easy to work directly with these probability distributions 
because they are over functions, instead of just real numbers. We provide the 
following lemma which often comes in handy and provides some guarantee that 
the assumptions lead to better mathematical properties.
\begin{lemma}\label{lemma:continuity-in-p-d}
	For any blockchain $(\vecc B, \vecc d, \vecc p)$ and regular demand $n, F$, 
	$\expect{T_j}$ is continuous in $\vecc p$ and $\vecc d$. 
	This holds \textbf{even} if we drop the last regularity assumption, namely 
	that $f(v_0 | h) > 0$.
\end{lemma}

\subsection{Compatible Blockchains}
Given these definition of supply (i.e., the blockchain and associated tiers) and 
demand, we can define the property that any such blockchain should 
have to be practical. Namely, that the expected demand of each tier fits within 
it.
\begin{definition}[Compatibility]\label{def:compatible}
	The blockchain $(\vecc B, \vecc d, \vecc p)$  is compatible with load $n$ 
	and $F$ if \Cref{eq:expected_size} is 
	satisfied for all tiers.
\end{definition}
The abstraction is that such a set of prices captures the essence of what 
happens at equilibrium, under any posted price scheme (e.g., similar to 
EIP-1559) and adapted to multiple tiers.
\begin{remark}
	This definition seems to exclude pricing systems where transactions 
	are \textbf{queued} into mempools and might wait several blocks before 
	being published. However, these are also captured: for these queues to have 
	finite length it is necessary for the `average' prices to satisfy 
	\ref{eq:expected_size}.
\end{remark}
Of course, for any load $n$, $F$ there are many blockchains $(\vecc B, \vecc d, 
\vecc p)$  satisfying 
Definition~\ref{def:compatible} (e.g., any with prices that are high enough to 
exclude all transactions). The following definition captures how EIP-1559 style 
price updates would translate into our steady-state model for a tiered 
mechanism.
\begin{definition}[EIP-1559 Stable Blockchain]\label{def:stable_prices}
	A blockchain $(\vecc B, \vecc d, \vecc p)$ is \emph{stable} for load $n, F$ if 
	it is compatible (recall Definition~\ref{def:compatible}) and for every 
	$j\in[k]$:
	\begin{equation}\label{eq:tight_prices}
		\frac{\expect{T_j}}{B_j} < 1\Rightarrow p_j = 0.
	\end{equation}	
\end{definition}
At first glance, it might seem as if Ethereum itself does not satisfy this 
definition. However, notice that for Ethereum, the target tier size (which is just 
the tier size in our notation) is not the physical maximum block size, but rather 
the target of having blocks be $50\%$ full. Under that prism, as long as the 
prices lead to low demand, they will keep decreasing until the eventually reach 
0, if necessary.

Given these notions of stability, we show that the previous distributional 
assumptions are indeed necessary.
\begin{observation}\label{obs:no_price}
	Suppose that $B_1 = B_2 = 1, d_1 = 0, d_2 = 1, n = 3$ and for $F$:
	\begin{itemize}
		\item With probability 1/3, $v_0$ is distributed uniformly in $[0,1]$ for 
		delay less 
		than $1$ and $0$ otherwise.
		\item With probability 2/3, $v_0$ is uniform in $[0,1]$ for any 
		delay.
	\end{itemize}
	If $p_1 \neq p_2$ then it has to be $p_1 > p_2$. Otherwise, all transactions 
	will 
	chose $B_1$, leaving $B_2$ empty. By Definition~\ref{def:stable_prices} 
	$p_2 = 0$, leading to a contradiction since we assumed $p_1 < p_2 = 0$ and 
	$p_1$ has to be non-negative. This leaves us with two options:
	\begin{itemize}
		\item $p_1 > p_2$. In this case, all transactions with $d=2$ will choose 
		to pay $p_2$. To achieve stability, we need that 		$p_2 > 1/2$. But then 
		$p_1 > p_2 > 1/2$ is too high to fill $B_1$.
		\item $p_1 = p_2$. In this case, $B_1$ will have higher demand than 
		$B_2$, as it is selected by transactions with $d=1$ and half of the 
		transactions with $d = 2$. Since $B_1 = B_2$, it is impossible to have prices 
		$p_1 = p_2$ that equally fill both blocks. Since the block that is not totally 
		filled needs to have a price equal to 0, $p_1 = p_2 > 0$ is not a possible 
		solution.
	\end{itemize}
\end{observation}
\subsection{Properties of Compatible and EIP-1559 Stable blockchains}
We need to show that the previously described tiered pricing mechanism does 
indeed converge to set of steady-state tier parameters. This requires laying some theoretical 
groundwork, before proceeding to establish that:
\begin{itemize}
	\item For fixed tier sizes and delays we can always find compatible and 
	EIP-1559 stable prices.
	\item These prices are unique.
	\item For fixed tier sizes, we can always find delays and prices implementing 
	a desired policy.
\end{itemize} 

The first step, also serving as a `sanity' check, for this model is all about 
Definitions \ref{def:compatible} and \ref{def:stable_prices}. Do blockchains 
such as Ethereum satisfy our definitions? The 
answer is positive. In Ethereum's case every parameter other than the price is 
set: there is just one tier with minimal delay.

\begin{proposition}
	Ethereum with EIP-1559 is compatible with any load. In particular, for any 
	$B$ and $n,F$ 
	there exists price $\vecc p$ such that $(\vecc B, \vecc d, \vecc p)$ is 
	compatible with $n, F$.
\end{proposition}
\begin{proof}
	For EIP-1559 we have that $B = B_1$ and $d_1 = 0$.  Clearly, for $p_1 = 0$ 
	we have that $\expect{T_1} = n$, by  Definition~\ref{def:demand_regular}. 
	Since for $p_1 > 1$ we have that $\expect{T_1} = 0$ (since the utility of any 
	transaction is at most 1) and $\expect{T_1}$ 
	is continuous in $p_1$, there exists some $p_1 > 0$ such that $\expect{T_j} = 
	B < n$, satisfying Definition~\ref{def:compatible}.
\end{proof}

This is a good time to highlight that for tiered blockchains, the combination of 
compatibility and EIP-1559 stability implies a more intuitive property.
\begin{proposition}\label{prop:all_tiers_full}
	For any blockchain $(\vecc B, \vecc d, \vecc p)$ that is compatible and 
	EIP-1559 stable with demand $n, F$, it holds that:
	$$
		\expect{T_j} = B_j
	$$
	for all $j \in [k]$.
\end{proposition}
\begin{proof}
	If for some tier $j$ it was $\expect{T_j} < B_j$, which is possible by 
	compatibility, EIP-1559 stability would imply that the price of that tier is 
	zero. However, in this case any transaction would have positive utility from 
	joining that tier. Since $n > \sum_{j \in [k]} B_j$ this tier would have demand 
	at least $n - (\sum_{\ell \in [k]} B_{\ell} - B_j) > B_j$ and therefore would not 
	be compatible.
\end{proof}
In addition, we show the simple fact that for high prices, the demand is zero.
\begin{lemma}
	For any blockchain $(\vecc B, \vecc d, \vecc p)$ and regular demand $n, F$,  
	if $p_j = 1$ then $\expect{T_j} = 0$ 
\end{lemma}
\begin{proof}
	\begin{align*}
		\expect{T_j} = n \cdot \prob{t \in T_j} &\le n \cdot \prob{v_0 \cdot h(d_j) - 
			p_j \ge 0}\\
		&\le n \cdot \prob{v_0 \ge p_j}\\
		&= 0,
	\end{align*}
	where the first inequality follows because for any transaction to enter tier $j$ 
	it's utility needs to be non-negative and the last equality by the last two 
	distributional assumptions (i.e., that $v_0$ is distributed in $[0,1]$ and has 
	zero probability of being equal to $1$).
\end{proof}
And a complementary result for low prices.
\begin{lemma}
	For any blockchain $(\vecc B, \vecc d, \vecc p)$ and regular demand $n, F$,  
	if $p_j = 0$ then $\sum_{j \in [k]}\expect{T_j} < B$ 
\end{lemma}
\begin{proof}
	If $p_j = 0$ then almost surely every transaction would have positive utility 
	in tier $j$, therefore would be included in some tier. The proof follows given 
	that the arrival rate is higher then the throughput: $n > B$.
\end{proof}
\begin{lemma}
	For any blockchain $(\vecc B, \vecc d, \vecc p)$ and demand $n, F$ where 
	$\vecc p$ is EIP-1559 stable, we have that 
	$$	p_j \ge p_{j+1} $$
	for all $1 \le j < k$.
\end{lemma}
\begin{proof}
	Suppose that for some $1\le j < k$ we have that $p_j < p_{j+1}$. In this case, 
	tier $j$ would have lower price and (by definition ) lower delay that tier 
	$j+1$. As such, by \Cref{def:trx_utility} we have that $u(i, j) > u(i,j + 1)$ for 
	any transaction $i$, leading to $T_{j+1} = \emptyset$. However, $p_{j+1} > 
	p_{j} \ge 0$, leading to a contradiction as an empty tier must necessarily 
	have zero price.
\end{proof}
\begin{remark}
	The previous result only holds if the price vector is EIP-1559. Compatibility 
	alone is not enough, as tier $j$ would be empty but it could be the case that 
	$T_j \le B_j$ if the $p_j$ is still high enough.
\end{remark}

We are now ready to prove our first result about the existence of prices. To 
avoid the cumbersome blockchain notation, we prove the following technical 
lemma about a continuous function $f$ which exhibits all salient properties 
of the expected demands $\expect{T_j}$ relative to the prices $\vecc p$, as 
proven with the previous lemmas.
Essentially, $f$ can be though of as the map from $\vecc p$ to $(\expect{T_1}, 
\expect{T_2}, \ldots, \expect{T_k})$.  The third property follows immediately: 
if a price of some tier increases, its demand decreases and the demand for other 
tiers cannot decrease.
\begin{lemma}\label{lemma:technical_existence}
	Let $f : [0, 1]^k \rightarrow \mathbb{R}_{\ge 0}^k$ be a function with the 
	following properties:
	\begin{enumerate}
		\item $f(x_1, x_2, \ldots, x_k)$ is uniformly continuous in $\vecc x$ and 
		decreasing 
		in $x_j$ for $j \in [k]$.
		\item If $x_j = 1$, then $f_j(\vecc x) = 0$, where $f_j$ is the $j-th$ 
		coordinate of $f(\vecc x)$.
		\item If $x_j = 0$ for some $j \in [k]$ then:
		$$
			\sum_{j \in [k]} f_j(\vecc x) > B.
		$$
		\item For $\vecc y = \vecc x + \delta \cdot \vecc e_j$ such that $\vecc x, 
		\vecc y < 1$ we have that $f_j(\vecc y) < f_j(\vecc x)$ and $f_\ell(\vecc y) 
		\ge f_\ell(\vecc x)$ for $\ell \neq j$.
	\end{enumerate}
	Then, for any $B_1, B_2, \ldots, B_k > 0$ such that $\sum_{j \in [k]} B_j = B$ 
	there exists some $\vecc x \in [0, 1]^n$ such that:
	$$
		f_j(\vecc x) = B_j > 0
	$$
	for all $j \in [k]$.
\end{lemma}

Putting it all together, we obtain our first result.
\begin{theorem}\label{thm:existence_of_prices}
	For any blockchain with fixed $\vecc B$ and $\vecc d$ and any $n, F$, 
	there exists $\vecc p$ such that $(\vecc B, \vecc d, \vecc p)$ is compatible 
	with the load and EIP-1559 stable.
\end{theorem}

Knowing that compatible and EIP-1559 stable $\vecc p$ always exist for any 
$(\vecc B, \vecc d)$, we are almost ready to prove that these prices are indeed 
unique. As a warmup to be used in the main proof, we show that if all prices 
increase, the 
aggregate demand cannot increase as well.
\begin{lemma}\label{lemma:dominated_price}
	Given $\vecc B, \vecc d$ and $n, F$ and two different price vectors $\vecc p 
	\ge \vecc p'$ (i.e., each price in $\vecc p$ is equal or higher than the 
	corresponding in $\vecc p'$), we have that
	$$
	\sum_{j \in [k]} T_j \le \sum_{j \in [k]} T_j'
	$$
	for all $j \in [k]$.
\end{lemma}
\begin{proof}
	Let $S$ be the set of valuations that would lead to positive utility if included 
	in their favorite tier of blockchain $(\vecc B, \vecc d, \vecc p')$. Clearly the 
	utility of all of them would increase if they are included in $(\vecc B, \vecc d, 
	\vecc p)$ instead in the same tier. Defining $S'$ similarly, we have that
	$S' \subseteq S \Rightarrow \prob{S'} \le \prob{S}$. Since $S$ (and $S'$) 
	can be partitioned into disjoint sets $S_j$ based on the preferred tier of each 
	transaction such that $\prob{S_j} = T_j$ multiplying by $n$ gives us the 
	claimed result.
\end{proof}
Showing that the prices are unique builds upon and refines this lemma.
\begin{theorem}\label{thm:uniqueness_of_prices}
	For any blockchain with fixed $\vecc B$ and $\vecc d$ and any $n, F$, if 
	there exists $\vecc p$ such that if $(\vecc B, \vecc d, \vecc p)$ is compatible 
	with the load and EIP-1559 stable, then $\vecc p$ is unique.
\end{theorem}
\begin{proof}
Suppose that there exist two price vectors $\vecc p \neq \vecc p'$ such that 
both  $(\vecc B, \vecc d, \vecc p)$ and  $(\vecc B, \vecc d, \vecc p')$ are 
compatible and EIP-1559 stable with demand $n,F$. Since  $\vecc p \neq 
\vecc p'$, we can partition the prices into two sets of indices:
$$
P^+ = \{j \in [k]\;|\; p'_j > p_j \}\text{ and } P^- = \{j \in [k]\;|\; p'_j \le p_j 
\}.
$$
We use $T_j$ and $T_j'$ for the tier contents of $(\vecc B, \vecc d, \vecc p)$ 
and $(\vecc B, \vecc d, \vecc p')$ respectively.

Suppose that $P^- = \emptyset$. We just need to show that for $(\vecc B, 
\vecc d, \vecc p')$ and demand $n,F$ we have
\begin{equation}\label{ineq:some_empty}
	\sum_{j \in [k]} \expect{T_j'} < \sum_{j \in [k]} B_j,
\end{equation}
which would immediately imply that the prices $\vecc p'$ are not EIP-1559 
stable, since they are all positive and some tiers are not full. 

To obtain a contradiction, we assume that \Cref{ineq:some_empty} holds 
with equality. Let 
$\delta = \min_j \{p_j' - p_j\;|\; j \in [k]\}$ and define a new price vector 
$\hat{\vecc p} = \vecc p + \delta$. By compatibility of prices and 
\Cref{lemma:dominated_price} we have
\begin{equation}\label{eq:same_stuff}
	\sum_{j \in [k]} B_j = \sum_{j \in [k]} \expect{T_j'} \ge \sum_{j \in [k]} 
	\expect{\hat T_j} \ge 
	\sum_{j \in [k]} \expect{T_j} = \sum_{j \in [k]} B_j. 
\end{equation}
Therefore, if \emph{all} prices increase by exactly $\delta$ the total demand 
stays the same. In fact, the demand for \emph{every} tier remains the same 
as 
in $(\vecc B, \vecc d, \vecc p)$  since the utilities are quasilinear and all 
prices 
changed by the same amount. Moreover, the preferred tier for any valuation 
$v$ is the same in $(\vecc B, \vecc d, \vecc p)$ and $(\vecc B, \vecc d, 
\vecc{ 
	\hat{p}})$   We will show that this violates our distributional 
assumptions.

Let $S_j$ be the subset of the support of $F$ containing the transactions that 
would prefer $T_j$ under $\vecc p$. In particular,
$$
S_j = \{(v_0 , h) \in F \;|\; v_0\cdot h(d_j) - p_1 \ge v_0 \cdot h(d_j') - p_j 
\text{ for 
	all } j' \neq j\}.
$$
We have that:
$$
\expect{T_j} = n\cdot \prob{S_j} \Rightarrow \prob{S_j} = \frac{B_j}{n}.
$$
Let $S$ be the subset of the support of $F$ containing the transactions that 
would join some tier under $\vecc p$:
$$
S = \{(v_0, h) \in F \;|\; v_0 \cdot h(d_j) - p_j \ge 0 \text{ for some } j \in 
[k]\}
$$
with $\hat S$ defined similarly. Of course, $S = \cup_{j \in [k]} S_j$ and 
every 
$S_j$ is disjoint.
Continuing from \Cref{eq:same_stuff}:
\begin{equation}
	\sum_{j \in [k]} \expect{\hat{T_j}} = \sum_{j \in [k]} \expect{T_j} 
	\Rightarrow \sum_{j \in [k]} \expect{\hat{T_j}}/n = \sum_{j \in [k]} 
	\expect{T_j}/n
	\Rightarrow \sum_{j \in [k]} \prob{\hat{S_j}}= \sum_{j \in [k]} 
	\prob{S_j}.
\end{equation}
Let $\mathbbm{1}_{t \in S}$ be the indicator variable for transaction $t$ 
joining $S$.  Since all $S_j$ are disjoint:
$$
\sum_{j \in [k]} \prob{\hat{S_j}}= \sum_{j \in [k]} \prob{S_j}
\Rightarrow 
\prob{\cup_{j \in [k]} \hat{S_j}}= \prob{\cup_{j \in [k]} S_j}
\Rightarrow
\expect{\mathbbm{1}_{t \in \hat{S}}} = \expect{\mathbbm{1}_{t \in S}}.
$$
To simplify this calculation we can use the law of total expectation, 
conditioning 
on the discount function $h$ of the valuation. Putting everything together, 
we 
have:
\begin{equation}\label{eq:contradiction}
	\sum_{j \in [k]} \expect{\hat{T_j}} = \sum_{j \in [k]} \expect{T_j}
	\Rightarrow
	\expect{\expect{\mathbbm{1}_{t \in \hat{S}}\;|\; h}}
	=
	\expect{\expect{\mathbbm{1}_{t \in S}\;|\; h}}
\end{equation}
where $h$ is a discount function. But then:
$$
\expect{\mathbbm{1}_{t \in S}\;|\; h} 
=
\prob{v_0 \ge V_h \;|\; h}
\quad\text{ and }\quad
\expect{\mathbbm{1}_{t \in \hat{S}}\;|\; h}
=
\prob{v_0 \ge \hat{V_h} \;|\; h}
$$
This follows because having conditioned on $h$, the valuation is 
increasing in $v_0$, therefore there is a minimum $V_h$ above which 
every $v_0$ would join some tier.  We know that 
for a transaction with $v_0 = \hat{V_h}$ (included in some tier $j$ under 
$\hat{\vecc p}$) we have
$$
\hat{V_h} \cdot h(d_j) - (p_j + \delta) \Rightarrow  \hat{V_h} \cdot h(d_j) - 
p_j 
\ge \delta > 0,
$$
so there exists some value $\hat{V_h} - \epsilon$ such that
$$
(\hat{V_h} - \epsilon)\cdot h(d_j)  - p_j > 0.
$$

Therefore under prices $\vecc p$, any $\hat{V_h} \ge v_0 \ge \hat{V_h} - 
\epsilon$ 
would be included in some tier for discount $h$:
\begin{align*}
	\prob{v_0 \ge V_h \;|\; h} &\ge \prob{v_0 \ge \hat{V_h}-\epsilon \;|\; h}\\
	&\ge
	\prob{\hat{V_h} > v_0 \ge \hat{V_h}-\epsilon \;|\; h}
	+
	\prob{v_0 \ge \hat{V_h}\;|\; h}\\ 
	&>
	\prob{v_0 \ge \hat{V_h}\;|\; h}.
\end{align*}
This implies that
$$
\expect{\mathbbm{1}_{t \in S}\;|\; h} 
>
\expect{\mathbbm{1}_{t \in \hat{S}}\;|\; h},
$$
which leads to a contradiction of \Cref{eq:contradiction} by taking the 
expected value of either side.

For the case where  $P^- \neq \emptyset$, clearly, the demand for any tier $j 
\in P^+$ can only be strictly less for $(\vecc B, \vecc d, \vecc p')$ than it is 
for 
$(\vecc B, \vecc d, \vecc p)$, as the price of those tiers increased while the 
remaining prices stayed the same or decreased. The proof proceeds using a 
slight modification of the previous analysis, with the added case that instead 
of 
the prices of $P^+$ not being EIP-1559 stable, there is a chance that some 
demand spilled over into the tier $P^-$ making the new prices incompatible 
with the demand.
\end{proof}

This result provides some guarantee that as long as the tiers and delays are 
mostly fixed (or at the very least, change slowly and the first line of defense 
against changing demand are the fluctuating prices), any sensible price update 
oscillations. 
\subsection{Steady-state of the Tiered Pricing Mechanism}
Under load $n, F$ and fixed $\vecc B$, the Tiered Pricing mechanism aims to find a 
blockchain $(\vecc B, \vecc d, \vecc p)$ such that:

\begin{align} \label{ineq:inclusive}
	d_{j+1} \ge \lambda_j \cdot d_j \text{ for } \lambda_j > 1,
	\text{ and } p_{j+1} \le \mu_j \cdot p_j \text{ for } \mu_j < 1
\end{align}

\ignore{
\begin{itemize} 
	\item $d_{j+1} \ge \lambda_j \cdot d_j$ for $\lambda_j > 1$
	\item $p_{j+1} \le \mu_j \cdot p_j$ for $\mu_j < 1$.
\end{itemize}
}
We show that it is always possible to find such delay and price combinations.
\begin{theorem}\label{thm:sec}
	For any $\vecc B$, factors $\lambda_j, \mu_j$ and regular demand $n, F$ there 
	exists a blockchain $(\vecc B, \vecc d, \vecc p)$ which is compatible, 
	EIP-1559 
	stable and satisfies Inequalities~\ref{ineq:inclusive}.
\end{theorem}
\begin{proof}
	Set $d_1 = 0$. Then, define
	$$
		d_{j+1} = \max\{\lambda_j \cdot d_j, \argmin_d \{\prob{h(d) < \mu _j 
		\cdot h(d_j) / 2} > 1 - 
		\delta \cdot \frac{1}{k-1}\cdot \frac{1}{ n}\}\}.
	$$
	Essentially, $d_{j+1}$ is sufficiently higher than $d_j$ so that the discount 
	factor incurred to almost all transactions is at least $\lambda_j / 2$ smaller 
	than for $d_j$.  We can capture this `almost all' property with the following 
	sets:
	 $$
	 H_{j+1} = \{(v_0, h) \in \text{supp}(F) \;|\;  h(d_{j+1}) < 
	\mu_j \cdot h(d_{j})\}
	$$
	and $H = \cap^k_{j \ge 1} H_j$. The probability that in a sample of $n$ 
	transactions we observe none with discount functions outside of the set $H$ 
	is:
	\begin{align*}
		\prob{\forall i \in [n] \; t_i \in H} &= 1 - \prob{\exists i \in [n] \text{ s.t. } 
		t_i \in 
		H^C}\\ 
		&\ge 1 - n \cdot \prob{\cup_{j \ge 1}^k H_j^C}\\
		&\ge 1 - n\cdot\sum_{j \ge 1}^k \delta / (k-1)/n = 1 - \delta.
	\end{align*}
	Therefore, the delays $d_{j+1}$ can be set high enough so that for with  
	probability at least $1-\delta$ all transactions in a sample have discount 
	functions such that consecutive tiers are at least $\mu_j$ times less 
	valuable.
	
	Let $\vecc p$ be the compatible and EIP-1559 stable prices guaranteed to 
	exist \Cref{thm:existence_of_prices}, for the delays $d_1, d_2, \ldots, d_k$ 
	that we just described. By \Cref{prop:all_tiers_full} we know that each tier 
	needs to be exactly full for these prices. Therefore, for any tier $j$ there must 
	exist a transaction with $h \in H$ and $v_0$ such that tier $j+1$ is 
	preferable to tier $j$ and its utility is positive. Formally:
	\begin{align*}
		v_0\cdot h(d_j) - p_j &< v_0 \cdot h(d_{j+1}) -p_{j+1}\\ 
		\Rightarrow
		v_0\cdot h(d_{j+1}) / \mu_j- p_j &< v_0 \cdot h(d_{j+1}) -p_{j+1}\\ 
		\Rightarrow
		v_0\cdot h(d_{j+1}) &< \frac{p_j - p_{j+1}}{1/\mu_j - 1}.
	\end{align*}
	Using that the utility is positive (i.e. $v_0\cdot h(d_{j+1}) - p_{j+1}> 0$ we 
	continue:
	\begin{align*}
		p_{j+1} &< \frac{p_j - p_{j+1}}{1/\mu_j - 1}\\
		\Rightarrow 
		p_{j+1} &< p_j \cdot \mu_j.
	\end{align*}
Therefore, for the appropriately chosen delays it is guaranteed that the prices 
will also follow our constraints.
\end{proof}

The above theorem can be rephrased to talk about the ability of the Tiered Pricing mechanism to implement policies in the class defined by Inequality~\ref{ineq:inclusive}.
\begin{corollary}
	For any regular demand $n, F$, any $k\geq B$, and any parameter set $\{(a_j,\lambda_j, \mu_j)\}_{i\in[k-1]}$, where $a_j\in(0,1],\lambda_j\in\mathbb{R}_{>1},\mu_j\in(0,1)$, for $j\in[k-1]$, and $\sum_{j}a_i =1$,  the Tiered Pricing mechanism with the same parameters and $tFreq:=\infty$ implements policies $P:=\{(a_i,d_i,p_i)\}_{i\in[k]}$ whose clauses satisfy Inequality~\ref{ineq:inclusive}. 
\end{corollary}

Unfortunately, there are many pairs of $\vecc p$ and $\vecc d$ satisfying 
these properties. For example, the delay of the last tier can always be increased 
(lowering its corresponding price), while maintaining our guarantees. However, 
our tiered pricing mechanism finds a local minimum solution, which does not 
unnecessarily inflate delays.

\begin{theorem}
	For any $\vecc B$, factors $\lambda_j, \mu_j$ and demand $n, F$ there 
	exists a blockchain $(\vecc B, \vecc d, \vecc p)$ which is compatible, 
	EIP-1559 stable and satisfies Inequalities~\ref{ineq:inclusive}.  
	Moreover, no delay can be independently lowered without sacrificing any of 
	our guarantees.
\end{theorem}
The proof is omitted, as it is similar \Cref{thm:existence_of_prices}.


\section{Experiments}


In this section, we present an experimental evaluation of the Tiered Pricing 
mechanism, in its complete generality as presented in Mechanism~\ref{alg:TP}. We set the 
maximum number of tiers to $k = 4$. The total size of the block fits $B = 120$ 
transactions 
and each tier can use a quarter of it. We update the delays every 25 blocks and 
the tier sizes every 100 blocks. In addition, we decrease the delays with 
probability $0.25$, add new tiers when the lower price is\footnote{To aid readability we allow prices to be bigger than $1$ in this section. This change does not affect any of the results of the previous section.} 2 and remove them if 
it drops below 1.

In the following experiment, the load changes dynamically over time. There are 
4 regions, each taking about a quarter of the overall execution timeline. All 
discount functions are of the form,
$u^{d - 1}$, where $u$ is an urgency parameter. For the 
first two regions, the value and discount factor of each transaction are 
correlated and drawn from  the following combinations:
\begin{itemize}
	\item With probability $0.2$: $v_0 \in \mathcal{N}(200, 10)$ and $u \in 
	\mathcal{U}[3,4]$.
	\item With probability $0.4$: $v_0 \in \mathcal{N}(50, 10)$ and $u \in 
	\mathcal{U}[2,3]$.
	\item With probability $0.4$: $v_0 \in \mathcal{N}(20, 2)$ and $u \in 
	\mathcal{U}[1,1.5]$.
\end{itemize}
where $\mathcal{N}[\mu,\sigma^2]$ is the normal distribution with mean $\mu$ and variance $\sigma^2$ and $\mathcal{U}[a,b]$ is the uniform distribution in the interval $[a,b]$.
In the last two regions, they are uncorrelated. 
\begin{itemize}
	\item  $v_0 \in \mathcal{U}(1, 90)$ and $u \in \mathcal{U}[1,2]$.
\end{itemize}
In addition the load varies: it starts by being slightly less than throughput, then 
4 times higher, 2 times higher and finally just below again. Observe how the 
mechanism responds to these changes.
\begin{figure}[H]
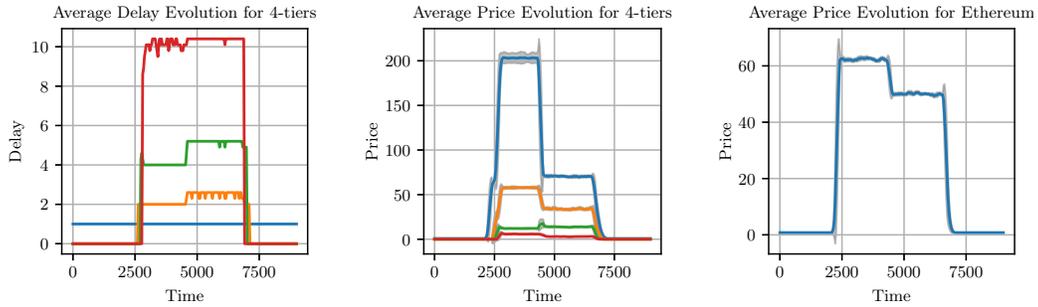

	\begin{center}
	\scalebox{0.7}{\input{figures/delays.pgf}}  

	\scalebox{0.7}{\input{figures/prices.pgf}}  

	\scalebox{0.7}{\input{figures/eth_prices.pgf}}  

	\end{center}
	\caption{The way the price of our tiered pricing mechanism evolves 
	compared to Ethereum. Notice that in the second region the urgent cluster 
	has been identified and priced accordingly, without raising the overall level of 
	prices too much.}
\end{figure}

\begin{figure}[H]
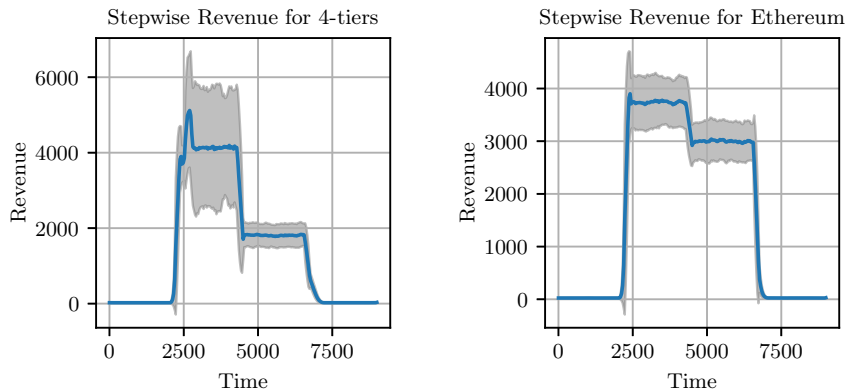

	\begin{center}
	\scalebox{0.9}{\input{figures/step_revenue.pgf}}  

	\scalebox{0.9}{\input{figures/eth_step_revenue.pgf}}  

	\end{center}
	\caption{The revenue of tiered pricing vs. Ethereum. Average revenue is 
	higher for tiered pricing in the second region, even though about $3/4$ of the 
	transactions are charged a lower price.}
\end{figure}

In general, it is hard to make generalizations about the performance of either 
system. We observe that in the short term, tiered pricing manages to identify 
and price urgent and non-urgent transactions accordingly, shifting the balance 
towards the latter while keeping a fair amount of the former. This comes at the 
cost of added complexity and convergence times after demand shocks. 
However, the long 
term growth of either ecosystem could be very different, based on the kind of 
traffic prioritized. Considering that a cooperation of high and low urgency 
applications are desirable, we propose that a tiered pricing solution could be 
better tailored to nurture this growth.


\bibliographystyle{plain}
\bibliography{tiers}

\appendix

\section{Omitted Proofs}

\ignore{
\begin{proof}[{\bf Proof of Theorem~\ref{th:EIPnotdiverse}}]
	Given the class parameters, we will describe a load $F$ for which EIP-1559 does not implement any policy in  $IP^k_{\{(a_i,\delta_i,\epsilon_i,\pi_i,\rho_i)\}_{i\in[k]}}$. Let $d_2 := \delta_i \cdot d_1 +\epsilon_1$ and $p_2 := (\pi_1-\rho_1)/p_1$. Assume $F$ outputs with probability $1/2$ valuation functions from two urgency classes, $h_1,h_2$. For the first class it holds that $\Pr[ v_0 \geq p_1  | h_1 ] = 0.99$ and that $h_1(d_2-1)=0$ (w.l.o.g, assume that $d_2-1 > d_2$). For the second class, it holds that $\Pr[ v_0 < (p_1+p_2)/2 | h_2 ] = 0.99\cdot a_2/(1-a_2)$ and $\Pr[ v_0 h_2(d_2) \geq p_2  | h_2 ] = 0.99$. 
			Due to our congestion assumption, i.e, that the number of functions output by $F$ is a lot larger that the available throughput $B$, it follows that EIP-1559 will set the cutoff price close to $p_1$. Moreover, on expectation less than 
			$$  0.99\cdot a_2/(1-a_2)/ (0.99\cdot a_2/(1-a_2)+0.99) < a_2$$
			transactions of the second type will make it to the blockchain. Given that transactions of the first type do not meet the second policy clause of any policy in $IP^k_{\{(a_i,\delta_i,\epsilon_i,\pi_i,\rho_i)\}_{i\in[k]}}$, it follows that EIP-1559 will not output enough transactions on expectation to meet the second policy clause, and the theorem follows. 
		\end{proof}
}
	
\begin{proof}[{\bf Proof of Lemma \ref{lemma:continuity-in-p-d}}]
		The proof only shows continuity for adjusting a specific $p_j$, not for the 
		vector $\vecc p$. However, the generalization is straightforward, 
		despite its cumbersome notation.
		
		For a given $\epsilon$, partition the set of possible discount functions into 
		two sets $H^+$ and $H^-$ defined as
		$$
		H^+ = \{h \;|\; \max_h f(v_0 | h) \le M\},
		$$
		picking $M \ge 0$ such that $\prob{H^+} \ge 1 - \epsilon / (6\cdot n)$ we 
		have that
		$$
		\max_{h \in H^+} f(v_0 | h) < M \le \infty
		$$
		If $H^+$ does not exist, this can only be because there are only finitely 
		many 
		options for $h$. In this case, we can define $\delta = \max_{h} f(v_0 | h) < 
		\infty$ without issue, since by our regularity assumption the marginal 
		density of $v_0$ is continuous for every $h$.
		
		Fix a  discount function $h \in H^+$. For this $h$, the 
		derivative with respect to $v_0$ of the  utility provided by every tier is 
		equal 
		to $h(d_j)$. Since $d_1 < d_2 < \ldots < d_k$ and $h$ is strictly decreasing, 
		we have that 
		$h(d_1) > h(d_2) > \ldots > h(d_k)$.  Therefore, we can partition the space 
		$[0,1]$ in $k$ closed-open 
		intervals $I_{j}$ (some of which could be empty) such that if $v_0 \in 
		I_{j}$ 
		then  $v_0 \cdot h(d_j) - p_j$ is maximized for tier $j$.
		
		Suppose that we adjust price $p_j$ by slightly increasing it by 
		$\delta_{p_j}$. 
		This has the effect of `shifting' the line $v_0 \cdot h(d_j) - p_j$ down. 
		Clearly, 
		if 
		$I_j = \emptyset$ then nothing changes. Otherwise, the two intervals 
		$I_{j+1} 
		= [a_{j+1}, b_{j+1}]$ and $I_{j-1} = [a_{j-1}, b_{j-1})$ grow and become 
		$I_{j+1} = [a_{j+1} - \delta_{j+1}, b_{j+1}]$  and $I_{j-1} = [a_{j+1}, b_{j+1} 
		+ 
		\delta_{j+1}]$. The interval $I_j$ by a total of $\delta_{j} = \delta_{j-1} + 
		\delta{j+1}$ 
		shrinks accordingly. This procedure is clearly continuous: the horizontal 
		displacement of the intersection of two lines (in $v_0$) of the form $v_0 
		\cdot h(d_j) - p_j - \delta_{p_j}$ and $v_0 \cdot h(d_{j+1}) - p_j$ is 
		$\delta_{p_j} / 
		(h(d_{j}) - h(d_{j+1})$, which is clearly continuous in $\delta_{p_j}$. Thus, 
		we can select
		can be selected such that $\delta_{j-1}, \delta_{j}, \delta_{j + 1} \le \delta$, 
		for any $\delta > 0$. We note that the intervals affected need not be 
		$I_{j-1}$ 
		and $I_{j+1}$. If these are empty, then the first non-empty interval from a 
		tier above and below is affected instead, but studying $I_{j-1}, I_{j+1}$ is 
		without loss of generality.
		
		We now focus on a specific tier $T_{\ell}$.  Consider two blockchains: the 
		original $(\vecc B, \vecc d, \vecc p)$ and the adjusted  $(\vecc B, \vecc d, 
		\vecc p')$, where $\vecc p'$ is equal $\vecc p$ except for $p'_j = p_j + 
		\delta_{p_j}$.  Given the previous observation, we know that for small 
		enough $\delta$ at most 3 more tiers can be affected, namely $j-1, j$ and 
		$j+1$. We continue the proof for the case $\ell = j$. The remaining two are 
		almost identical. We have that $\prob{t_i \in T_j | h} = 
		\prob{v_0 \in I_j | h}$. If the price of tier $j$ increases by $\delta_{p_j}$, 
		the 
		interval $I_j$ will 
		decrease by at most $\delta$ as well (setting $\delta_{p_j}$ as above). 
		Thus 
		if $I_j = [a_j, b_j)$, for $p_j' = p_j + 
		\delta$ (and $I_j', T_j'$ defined accordingly for the modified blockchain) 
		Continuing and setting $\delta = \epsilon / (3 \cdot M\cdot n)$ we have 
		that:
		\begin{align*}
			\prob{t_i \in T_j' \;|\; h} &= \prob{v_0 \in I_j' | h}\\ 
			&=\prob{v_0 \in I_j' | h} + \prob{v_0 \in I_j \setminus I_j' | h}
			- \prob{v_0 \in I_j \setminus I_j' | h}\\
			&=\prob{v_0 \in I_j | h} - \prob{v_0 \in I_j \setminus I_j' | h}\\
			&=\prob{v_0 \in I_j | h} - \int_{I_j \setminus I_j'} f(v_0 | h)dv_0\\
			&\ge\prob{v_0 \in I_j | h} -  M \cdot \delta\\
			&= \prob{t_i \in T_j | h} - \epsilon / (3 \cdot n),
		\end{align*}
		since we know that $I_j \setminus I_j'$ are two intervals of length at most 
		$\delta$ and $f(v_0 | h) \le M$ for $h \in H^+$. At the same time, we have 
		the straightforward bound
		$$
		\prob{t_i \in T_j' \;|\; h} \le \prob{t_i \in T_j \;|\; h}.
		$$
		
		So far we have shown that for a specific $h \in H^+$, the resulting tier 
		inclusion probabilities are affected in continuous fashion by changing 
		$p_j$. 
		However, the $\delta_{p_j}$ 
		selected might not apply for all $h \in H^+$. We define the following set
		$$
		H(\delta_{p_j}, \delta) = \left\{h \;|\;  \forall j:\;\frac{\delta_{p_j}}{h(d_j) 
			- 
			h(d_{j+1})} < \delta \right\},
		$$
		containing the discount functions for which a change $\delta_{p_j}$ of 
		price 
		$p_j$ displaces all previously mentioned intervals $I_j$ by at most 
		$\delta$. 
		By continuity and strict monotonicity of every $h$, for any $\delta$ we 
		can 
		set $\delta_{p_j}$ small enough such that $\prob{H(\delta_{p_j}, 
		\delta)}$ 
		is a close to 1 as needed.
		
		We set $\delta = \epsilon / (2\cdot M\cdot n)$ and select $\delta_{p_j}$ 
		accordingly, such that $\prob{H(\delta_{p_j}, \delta)} \le \epsilon / 
		(6\cdot 
		n)$. 
		Putting 
		everything together, we have that:
		\begin{align*}
			\expect{T_j'} / n & = \prob{t_i \in T_j'}\\ 
			&= \prob{t_i \in T_j' | h \in H(\delta_{p_j}, \delta) \cap H^+}\cdot 
			\prob{H(\delta_{p_j}, \delta) \cap H^+}\\
			&\qquad+ \prob{t_i 
				\in T_j' | h \neq H(\delta_{p_j}, \delta) \cup H^+}\cdot 
			\prob{(H(\delta_{p_j}, \delta) \cup H^+)^C}\\
			&\le \prob{t_i \in T_j' | h \in H(\delta_{p_j}, \delta) \cup H^+} + 
			\prob{(H(\delta_{p_j}, \delta) \cup H^+)^C}\\
			&\le \prob{t_i \in T_j |  h \in H(\delta_{p_j}, \delta) \cup H^+} + 
			\epsilon/(3\cdot n) + \epsilon /(3\cdot n)\\
			&\le \prob{t_i \in T_j} + \epsilon/(3\cdot n) +
			\epsilon/(3\cdot n) + \epsilon / (3\cdot n)\\
			&\le \expect{T_j} / n  + \epsilon / n.
		\end{align*}
		where in the last step we took an upper bound by assuming that $\prob{t_i 
			\in T_j |  h \notin H(\delta_{p_j}, \delta) \cup H^+} = 1$, thus removing 
			the 
		conditional.
		
		The case for a lower price $p_j$ is very similar, concluding the proof.
	\end{proof}

\begin{proof}[{\bf Proof of Lemma \ref{lemma:technical_existence}}]
	We describe an iterative process that will converge to the desired point.
	\begin{itemize}
		\item Initially, set $\vecc x_0 = 0$.
		\item Then repeat:
		\begin{itemize}
			\item Let $\ell = \argmax_{j \in [k]} \{f_j(x_{t}) - B_j \}.$
			\item Set $\vecc x_{t+1} = \vecc x_{t} + \vecc e_\ell \cdot \delta$ for 
			some $\delta > 
			0$ such that $f_j(\vecc x_{t+1}) = B_j$.
		\end{itemize}
	\end{itemize}
	This process used the first two properties of $f$ to ensure that it is always 
	possible to find $x_{t+1}$ from $x_{t}$. Since for all $j$ the $x_j^t$ keep 
	increasing as $t$ grows,  but can never exceed 1 (by our second property), 
	the procedure converges to $\lim_{t \rightarrow \infty} \vecc x^t = \vecc 
	x^\star$.
	
	Suppose that $x_j^\star > 0$ for all $j \in [k]$. In this case, it is guaranteed 
	that $f_j(\vecc x^\star) \ge B_j$ for all $j$. To see this consider an update 
	$\vecc x_{t+1} = \vecc x_{t} + e_\ell \cdot \delta$ happening at time $t$ 
	where $\vecc x_t > 0$. By our last property, it is clear that $f_j(\vecc 
	x_{t+1}) \ge f_j(\vecc x_{t}) \ge B_j$ for $j \neq \ell$ and $f_\ell(\vecc 
	x_{t+1}) = B_\ell$. Therefore, for any $t' > t$ we have the invariant that 
	$f_j(\vecc x_{t'}) \ge B_j$ for all $j \in [k]$.
	
	We proceed to use a proof by contradiction. Suppose that for some $j \in [k]$, 
	we have that $f_j(\vecc x^\star) - B_j > 0$. Given the convergence to $\vecc 
	x^\star$, we know that for any $\delta > 0$ there exist some $T$ such 
	that  for all $t > T$, $|\vecc x_{t} - \vecc x^\star|_{\infty} < \delta$. In 
	addition, by uniform continuity of $f$  there exists a 
	$\hat \delta$ such that for any $|\vecc x - \vecc y|_{\infty} < \hat \delta$ 
	we have 
	that 
	$|f_j(\vecc x) - f_j(\vecc y)| < f_j(\vecc x^\star) - B_j$.
	
	There are two cases:
	\begin{itemize}
		\item Suppose that for any $T$, there exists $t > T$ such that at time 
		$\vecc x_{t+1} = \vecc x_{t} + a \cdot e_j$ for some $a$. Pick some $t$ 
		large enough such that $|\vecc x_{t+1} - \vecc x^\star|_{\infty} \le \hat 
		\delta$. But then we would have 
		$$|f_j(\vecc x_{t+1}) - f_j(\vecc x^\star)| < |B_j - f_j(\vecc x^\star)|
		\Rightarrow
		|B_j - f_j(\vecc x^\star)| < |B_j - f_j(\vecc x^\star)|,$$ leading to a 
		contradiction. We used the fact that after an update $f_j(\vecc x_{t+1}) = 
		B_j$.
		\item Suppose that there exist a $T$ such that there exists no $t > T$ for 
		which we have an 
		update $\vecc x_{t+1} = \vecc x_{t} + a \cdot e_j$. Since we always select 
		the $\ell$ for which $f_\ell(\vecc x_t) - B_\ell$ is maximized in order to 
		update $x_t$, we would have that $f_\ell(\vecc x_{t}) - B_\ell > f_j(\vecc 
		x_{t}) - B_j$. The proof now continues as the previous case, with the 
		discontinuity arising from reducing $f_\ell(\vecc x_t)$ too suddenly.
	\end{itemize}
	
	To conclude the proof, we need to consider the last case, that for some $j$ we 
	have $x^\star_j = 0$. If this indeed happened, then by our third property we 
	know that there exist some $j$ so that $f_j(\vecc x^\star) > B_j$. This case is 
	almost identical to the ones studied previously, since this would imply  that 
	our process kept performing `large' updates for arbitrarily large $t$.
\end{proof}

\end{document}